\documentclass[11pt]{article}
\usepackage{hyperref}
\usepackage{times}  % Fonts
\usepackage{mathpazo}
\usepackage{amssymb,amsmath,amsthm}
\usepackage{epsfig}

\newtheorem{theorem}{Theorem}[section]
\newtheorem{proposition}[theorem]{Proposition}
\newtheorem{definition}[theorem]{Definition}

\newtheorem{lemma}[theorem]{Lemma}
\newtheorem{conjecture}[theorem]{Conjecture}

\newtheorem{corollary}[theorem]{Corollary}

\newtheorem{remark}[theorem]{Remark}

\newcommand{\qedsymb}{\hfill{\rule{2mm}{2mm}}}
\renewenvironment{proof}[1][]{\begin{trivlist}
\item[\hspace{\labelsep}{\bf\noindent Proof#1:\/}] }{\qedsymb\end{trivlist}}

\def\calS{{\cal S}}
\def\calM{{\cal M}}
\def\calP{{\cal P}}

\def\R{\mathbb{R}}

\newcommand\Prob[2]{{\Pr_{#1}\left[ {#2} \right]}}

\newcommand\Kneser[3]{K^<(#1,#2,#3)}

\newcommand{\eps}{\epsilon}
\renewcommand{\epsilon}{\varepsilon}

\newcommand{\minrank}{\mathop{\mathrm{minrk}}}

\newcommand{\Fset}{\mathbb{F}}         % The integers
\begin{document}

\title{{\bf On Minrank and Forbidden Subgraphs}}

\author{
Ishay Haviv\thanks{School of Computer Science, The Academic College of Tel Aviv-Yaffo, Tel Aviv 61083, Israel.
}
}

\date{}

\maketitle

\begin{abstract}
The {\em minrank} over a field $\Fset$ of a graph $G$ on the vertex set $\{1,2,\ldots,n\}$ is the minimum possible rank of a matrix $M \in \Fset^{n \times n}$ such that $M_{i,i} \neq 0$ for every $i$, and $M_{i,j}=0$ for every distinct non-adjacent vertices $i$ and $j$ in $G$.
For an integer $n$, a graph $H$, and a field $\Fset$, let $g(n,H,\Fset)$ denote the maximum possible minrank over $\Fset$ of an $n$-vertex graph whose complement contains no copy of $H$.
In this paper we study this quantity for various graphs $H$ and fields $\Fset$.
For finite fields, we prove by a probabilistic argument a general lower bound on $g(n,H,\Fset)$, which yields a nearly tight bound of $\Omega(\sqrt{n}/\log n)$ for the triangle $H=K_3$.
For the real field, we prove by an explicit construction that for every non-bipartite graph $H$, $g(n,H,\R) \geq n^\delta$ for some $\delta = \delta(H)>0$.
As a by-product of this construction, we disprove a conjecture of Codenotti, Pudl\'ak, and Resta.
The results are motivated by questions in information theory, circuit complexity, and geometry.
\end{abstract}

\section{Introduction}

An $n \times n$ matrix $M$ over a field $\Fset$ is said to {\em represent} a digraph $G=(V,E)$ with vertex set $V = \{1,2,\ldots,n\}$ if $M_{i,i} \neq 0$ for every $i$, and $M_{i,j}=0$ for every distinct $i,j$ such that $(i,j) \notin E$. The {\em minrank} of $G$ over $\Fset$, denoted ${\minrank}_\Fset(G)$, is the minimum possible rank of a matrix $M \in \Fset^{n \times n}$ representing $G$. The definition is naturally extended to (undirected) graphs by replacing every edge with two oppositely directed edges.
It is easy to see that for every graph $G$ the minrank parameter is sandwiched between the independence number and the clique cover number, that is, $\alpha(G) \leq {\minrank}_\Fset(G) \leq \chi(\overline{G})$.
For example, ${\minrank}_\Fset(K_n)=1$ and ${\minrank}_\Fset(\overline{K_n})=n$ for every field $\Fset$.
The minrank parameter was introduced by Haemers in 1979~\cite{Haemers79}, and since then has attracted a significant attention motivated by its various applications in information theory and in theoretical computer science (see, e.g.,~\cite{Haemers81,BBJK06,Valiant92,Riis07,PudlakRS97,HavivL13,ChlamtacH14}).

In this work we address the extremal behavior of the minrank parameter of $n$-vertex graphs whose complements are free of a fixed forbidden subgraph.
For two graphs $G$ and $H$, we say that $G$ is {\em $H$-free} if $G$ contains no subgraph, induced or not, isomorphic to $H$.
For an integer $n$, a graph $H$, and a field $\Fset$, let $g(n,H,\Fset)$ denote the maximum of ${\minrank}_\Fset(G)$ taken over all $n$-vertex graphs $G$ whose complement $\overline{G}$ is $H$-free.
Our purpose is to study the quantity $g(n,H,\Fset)$ where $H$ and $\Fset$ are fixed and $n$ is growing.

\subsection{Our Contribution}

We provide bounds on $g(n,H,\Fset)$ for various graph families and fields.
We start with a simple upper bound for a forest $H$.

\begin{proposition}\label{prop:forestIntro}
For every integer $n$, a field $\Fset$, and a nontrivial forest $H$ on $h$ vertices,
\[g(n,H,\Fset) \leq h-1.\]
Equality holds whenever $H$ is a tree and $n \geq h-1$.
\end{proposition}

We next provide a general lower bound on $g(n,H,\Fset)$ for a graph $H$ and a finite field $\Fset$.
To state it, we need the following notation.
For a graph $H$ with $h \geq 3$ vertices and $f \geq 3$ edges define $\gamma(H) = \frac{h-2}{f-1}$ and $\gamma_0(H) = \min_{H'}{\gamma(H')}$, where the minimum is taken over all subgraphs $H'$ of $H$ with at least $3$ edges.

\begin{theorem}\label{thm:IntroComp}
For every graph $H$ with at least $3$ edges there exists $c=c(H)>0$ such that for every integer $n$ and a finite field $\Fset$,
\[g(n,H,\Fset) \geq c \cdot \frac{n^{1-\gamma_0(H)}}{\log (n \cdot |\Fset|)} .\]
\end{theorem}
%\noindent

Note that for every finite field $\Fset$, the quantity $g(n,H,\Fset)$ grows with $n$ if and only if $H$ is not a forest.
Indeed, if $H$ is a forest then $g(n,H,\Fset)$ is bounded by some constant by Proposition~\ref{prop:forestIntro}, whereas otherwise $H$ satisfies $\gamma_0(H)<1$ and thus, by Theorem~\ref{thm:IntroComp}, $g(n,H,\Fset) \geq \Omega(n^\delta)$ for some $\delta = \delta(H)>0$.
Note further that for the case $H=K_3$, which is motivated by a question in information theory (see Section~\ref{sec:applications}),
Theorem~\ref{thm:IntroComp} implies that
\begin{eqnarray}\label{eq:K_3}
g(n,K_3,\Fset) \geq \Omega  \Big ( \frac{\sqrt{n}}{\log n} \Big )
\end{eqnarray}
for every fixed finite field $\Fset$.
This is tight up to a $\sqrt{\log n}$ multiplicative term (see Proposition~\ref{prop:K_3}).

Theorem~\ref{thm:IntroComp} is proved by a probabilistic argument based on the Lov\'asz Local Lemma~\cite{LLL75}.
The proof involves an approach of Spencer~\cite{Spencer77} to lower bounds on off-diagonal Ramsey numbers and a technique of Golovnev, Regev, and Weinstein~\cite{Golovnev0W17} for estimating the minrank of random graphs.

As our final result, we show that for every non-bipartite graph $H$ there are $H$-free graphs with low minrank over the real field $\R$.

\begin{theorem}\label{thm:IntroNonBi}
For every non-bipartite graph $H$ there exists $\delta=\delta(H)>0$ such that for every sufficiently large integer $n$, there exists an $n$-vertex $H$-free graph $G$ such that ${\minrank}_\R(G) \leq n^{1-\delta}$.
\end{theorem}
\noindent
This theorem is proved by an explicit construction from the family of generalized Kneser graphs, whose minrank was recently studied in~\cite{Haviv18}.
It is known that every $n$-vertex graph $G$ satisfies
\begin{eqnarray}\label{eq:minrk_comp}
{\minrank}_\Fset(G) \cdot {\minrank}_\Fset(\overline{G}) \geq n
\end{eqnarray}
for every field $\Fset$ (see, e.g.,~\cite[Remark~2.2]{Peeters96}).
This combined with the graphs given in Theorem~\ref{thm:IntroNonBi} implies the following (explicit) lower bound on $g(n,H,\R)$ for non-bipartite graphs $H$.

\begin{corollary}\label{cor:IntroNonBi}
For every non-bipartite graph $H$ there exists $\delta=\delta(H)>0$ such that for every sufficiently large integer $n$,
$g(n,H,\R) \geq n^{\delta}$.
\end{corollary}
\noindent
As another application of Theorem~\ref{thm:IntroNonBi}, we disprove a conjecture of Codenotti, Pudl\'ak, and Resta~\cite{CodenottiPR00} motivated by Valiant's approach to circuit lower bounds~\cite{Valiant77} (see Section~\ref{sec:applications}).

\subsection{Applications}\label{sec:applications}

The study of the quantity $g(n,H,\Fset)$ is motivated by questions in information theory, circuit complexity, and geometry.
We gather here several applications of our results.

\paragraph{Shannon Capacity.}
For an integer $k$ and a graph $G$ on the vertex set $V$, let $G^k$ denote the graph on the vertex set $V^k$ in which two distinct vertices $(u_1,\ldots,u_k)$ and $(v_1,\ldots,v_k)$ are adjacent if for every $1 \leq i \leq k$ it holds that $u_i$ and $v_i$ are either equal or adjacent in $G$.
The Shannon capacity of a graph $G$, introduced by Shannon in 1956~\cite{Shannon56}, is defined as the limit $c(G) = \lim_{k \rightarrow \infty}{(\alpha(G^k))^{1/k}}$.
This graph parameter is motivated by information theory, as it measures the zero-error capacity of a noisy communication channel represented by $G$.
An upper bound on $c(G)$, known as the Lov\'asz $\vartheta$-function, was introduced in~\cite{Lovasz79}, where it was used to show that the Shannon capacity of the cycle on $5$ vertices satisfies $c(C_5)=\sqrt{5}$, whereas its independence number is $2$.
Haemers introduced the minrank parameter in~\cite{Haemers79,Haemers81} and showed that it forms another upper bound on $c(G)$ and that for certain graphs it is tighter than the $\vartheta$-function.
In general, computing the Shannon capacity of a graph seems to be a very difficult task, and its exact value is not known even for small graphs such as the cycle on $7$ vertices.

The question of determining the largest possible Shannon capacity of a graph with a given independence number is widely open.
In fact, it is not even known if the Shannon capacity of a graph with independence number $2$ can be arbitrarily large~\cite{AlonPowers02}.
Interestingly, Erd\"{o}s, McEliece, and Taylor~\cite{ErdosMT71} have shown that this question is closely related to determining an appropriate multicolored Ramsey number, whose study in~\cite{XiaodongZER04} implies that there exists a graph $G$ with $\alpha(G)= 2$ and $c(G)> 3.199$.
A related question, originally asked by Lov\'asz, is that of determining the maximum possible $\vartheta$-function of an $n$-vertex graph with independence number $2$. This maximum is known to be $\Theta(n^{1/3})$, where the upper bound was proved by Kashin and Konyagin~\cite{KasKon81,Kon81}, and the lower bound was proved by Alon~\cite{Alon94} via an explicit construction.
Here we consider the analogue question of determining the maximum possible minrank, over any fixed finite field $\Fset$, of an $n$-vertex graph with independence number $2$.
Since the latter is precisely $g(n,K_3,\Fset)$, our bound in~\eqref{eq:K_3} implies that the minrank parameter is weaker than the $\vartheta$-function with respect to the general upper bounds that they provide on the Shannon capacity of $n$-vertex graphs with independence number $2$.

\paragraph{The Odd Alternating Cycle Conjecture.}
In 1977, Valiant~\cite{Valiant77} proposed the matrix rigidity approach for proving superlinear circuit lower bounds, a major challenge in the area of circuit complexity.
Roughly speaking, the rigidity of a matrix $M \in \Fset^{n \times n}$ for a constant $\eps>0$ is the minimum number of entries that one has to change in $M$ in order to reduce its rank over $\Fset$ to at most $\eps \cdot n$. Valiant showed in~\cite{Valiant77} that matrices with large rigidity can be used to obtain superlinear lower bounds on the size of logarithmic depth arithmetic circuits computing linear transformations.
With this motivation, Codenotti, Pudl\'ak, and Resta~\cite{CodenottiPR00} raised in the late nineties the Odd Alternating Cycle Conjecture stated below, and proved that it implies, if true, that certain explicit circulant matrices have superlinear rigidity.
By an alternating odd cycle we refer to a digraph which forms a cycle when the orientation of the edges is ignored, and such that the orientation of the edges alternates with one exception.
\begin{conjecture}[The Odd Alternating Cycle Conjecture~\cite{CodenottiPR00}]\label{conj:alternating}
For every field $\Fset$ there exist $\eps >0$ and an odd integer $\ell$ such that every $n$-vertex digraph $G$ with ${\minrank}_\Fset (G)  \leq \eps \cdot n$ contains an alternating cycle of length $\ell$.
\end{conjecture}

Codenotti et al.~\cite{CodenottiPR00} proved that the statement of Conjecture~\ref{conj:alternating} does not hold for $\ell=3$ over any field $\Fset$. Specifically, they provided an explicit construction of $n$-vertex digraphs $G$, free of alternating triangles, with ${\minrank}_\Fset (G)  \leq O(n^{2/3})$ for every field $\Fset$. For the undirected case, which is of more interest to us, a construction of~\cite{CodenottiPR00} implies that there are $n$-vertex triangle-free graphs $G$ such that ${\minrank}_\Fset(G) \leq O(n^{3/4})$ for every field $\Fset$ (see~\cite[Section~4.2]{BlasiakKL13} for a related construction over the binary field as well as for an application of such graphs from the area of index coding). Note that this yields, by~\eqref{eq:minrk_comp}, that $g(n,K_3,\Fset) \geq \Omega(n^{1/4})$.
In contrast, for the real field and the cycle on $4$ vertices, it was shown in~\cite{CodenottiPR00} that every $n$-vertex $C_4$-free graph $G$ satisfies ${\minrank}_\R (G) > \frac{n}{6}$.
Yet, the question whether every $n$-vertex digraph with sublinear minrank contains an alternating cycle of odd length $\ell \geq 5$ was left open in~\cite{CodenottiPR00} for every field.
Our Theorem~\ref{thm:IntroNonBi} implies that for every odd $\ell$ there are (undirected) $C_\ell$-free graphs $G$ with sublinear ${\minrank}_\R(G)$, and in particular disproves Conjecture~\ref{conj:alternating} for the real field $\R$.

\paragraph{Nearly Orthogonal Systems of Vectors.}
A system of nonzero vectors in $\R^m$ is said to be nearly orthogonal if any set of three vectors of the system contains an orthogonal pair.
It was proved by Rosenfeld~\cite{Rosenfeld91} that every such system has size at most $2m$.
An equivalent way to state this, is that every $n$-vertex graph represented by a real positive semidefinite matrix of rank smaller than $\frac{n}{2}$ contains a triangle.
Note that the positive semidefiniteness assumption is essential in this result, as follows from the aforementioned construction of~\cite{CodenottiPR00} of $n$-vertex triangle-free graphs $G$ with ${\minrank}_\R(G) \leq O(n^{3/4})$.

A related question was posed by Pudl\'ak in~\cite{Pudlak02}.
He proved there that for some $\eps >0$, every $n$-vertex graph represented by a real positive semidefinite matrix of rank at most $\eps \cdot n$ contains a cycle of length $5$. Pudl\'ak asked whether the assumption that the matrix is positive semidefinite can be omitted.
Our Theorem~\ref{thm:IntroNonBi} applied to $H = C_5$ implies that there are $C_5$-free graphs $G$ with sublinear ${\minrank}_\R(G)$, and thus answers this question in the negative.

\subsection{Outline}
The rest of the paper is organized as follows.
In Section~\ref{sec:forest} we present the simple proof of Proposition~\ref{prop:forestIntro}.
In Section~\ref{sec:g_comp} we provide some background on sparse-base matrices from~\cite{Golovnev0W17} and then prove Theorem~\ref{thm:IntroComp}.
In the final Section~\ref{sec:non-bip}, we prove Theorem~\ref{thm:IntroNonBi}.

\section{Forests}\label{sec:forest}

In this section we prove Proposition~\ref{prop:forestIntro}.
We use an argument from one of the proofs in~\cite{AlonKS05}.

\begin{proof}[ of Proposition~\ref{prop:forestIntro}]
Fix a nontrivial $h$-vertex forest $H$ and a field $\Fset$.
It suffices to consider the case where $H$ is a tree, as otherwise $H$ is a subgraph of some $h$-vertex tree $H'$, and since every $H$-free graph is also $H'$-free, we have $g(n,H,\Fset) \leq g(n,H',\Fset)$.

Our goal is to show that every $n$-vertex graph $G$ whose complement $\overline{G}$ is $H$-free satisfies ${\minrank}_\Fset(G) \leq h-1$.
Let $G$ be such a graph.
We claim that $\overline{G}$ is $(h-2)$-degenerate, that is, every subgraph of $\overline{G}$ contains a vertex of degree at most $h-2$. Indeed, otherwise $\overline{G}$ has a subgraph $G'$ all of whose degrees are at least $h-1$, and one can find a copy of $H$ in $G'$ as follows: First identify an arbitrary vertex of $G'$ with an arbitrary vertex of $H$, and then iteratively identify a vertex of $G'$ with a leaf added to the being constructed copy of the tree $H$. The process succeeds since $H$ has $h$ vertices and every vertex of $G'$ has degree at least $h-1$.
As is well known, the fact that $\overline{G}$ is $(h-2)$-degenerate implies that $\overline{G}$ is $(h-1)$-colorable, so we get that ${\minrank}_\Fset(G) \leq \chi(\overline{G}) \leq h-1$, as required.

We finally observe that the bound is tight whenever $H$ is a tree and $n \geq h-1$.
Indeed, let $G$ be the $n$-vertex complete $\lceil \frac{n}{h-1} \rceil$-partite graph, that has $h-1$ vertices in each of its parts, except possibly one of them.
Its complement $\overline{G}$ is a disjoint union of cliques, each of size at most $h-1$, and is thus $H$-free.
Since $\alpha(G) = \chi(\overline{G})=h-1$, it follows that ${\minrank}_\Fset(G) = h-1$ for every field $\Fset$, completing the proof.
\end{proof}

\section{A General Lower Bound on $g(n,H,\Fset)$}\label{sec:g_comp}

In this section we prove Theorem~\ref{thm:IntroComp} and discuss its tightness for $H=K_3$.
We start with some needed preparations.

\subsection{Lov\'{a}sz Local Lemma}

The Lov\'{a}sz Local Lemma~\cite{LLL75} stated below is a powerful probabilistic tool in Combinatorics (see, e.g.,~\cite[Chapter~5]{AlonS16}).
We denote by $[N]$ the set of integers from $1$ to $N$.

\begin{lemma}\label{lemma:lll}[Lov\'{a}sz Local Lemma~\cite{LLL75}]
Let $A_1,\ldots, A_N$ be events in an arbitrary probability space.
A digraph $D = (V,E)$ on the vertex set $V = [N]$ is called a dependency digraph for the events $A_1,\ldots, A_N$ if for every $i \in [N]$, the event $A_i$ is mutually independent of the events $A_j$ with $j \neq i$ and $(i,j) \notin E$.
Suppose that $D=(V,E)$ is a dependency digraph for the above events and suppose that there are real numbers $x_1,\ldots,x_N \in [0,1)$ such that \[\Prob{}{A_i} \leq x_i \cdot \prod_{(i,j) \in E}{(1-x_j)}\] for all $i \in [N]$.
Then, with positive probability no event $A_i$ holds.
\end{lemma}

\subsection{Sparse-base Matrices}\label{sec:GRW}

Here we review several notions and lemmas due to Golovnev, Regev, and Weinstein~\cite{Golovnev0W17}.
For a matrix $M$ over a field $\Fset$, let $s(M)$ denote its sparsity, that is, the number of its nonzero entries.
We say that a matrix $M$ over $\Fset$ with rank $k$ contains an $\ell$-sparse column (row) basis if $M$ contains $k$ linearly independent columns (rows) with a total of at most $\ell$ nonzero entries.
We first state a lemma that provides an upper bound on the number of matrices with sparse column and row bases.

\begin{lemma}[\cite{Golovnev0W17}]\label{lemma:size_M}
The number of rank $k$ matrices in $\Fset^{n \times n}$ that contain $\ell$-sparse column and row bases is at most $(n \cdot |\Fset|)^{6\ell}$.
\end{lemma}

The following lemma relates the sparsity of a matrix with nonzero entries on the main diagonal to its rank.

\begin{lemma}[\cite{Golovnev0W17}]\label{lemma:sparsity_M}
For every rank $k$ matrix $M \in \Fset^{n \times n}$ with nonzero entries on the main diagonal,
\[s(M) \geq \frac{n^2}{4k}.\]
\end{lemma}

We also need the following notion. An {\em $(n,k,s,\ell)$-matrix} over a field $\Fset$ is a matrix in $\Fset^{n \times n}$ of rank $k$ and sparsity $s$ that contains $\ell$-sparse column and row bases and has nonzero entries on the main diagonal. Note that by Lemma~\ref{lemma:sparsity_M}, an $(n,k,s,\ell)$-matrix exists only if $s \geq \frac{n^2}{4k}$.
For integers $n,k,s'$ and a field $\Fset$ (which will always be clear from the context), let $\calM_{n,k}^{(s')}$ be the collection that consists of all $(n',k',s',\frac{2s'k'}{n'})$-matrices over $\Fset$ for all $n' \in [n]$ and $k' \in [k]$ such that $\frac{k'}{n'} \leq \frac{k}{n}$.
This collection is motivated by the following lemma.

\begin{lemma}[\cite{Golovnev0W17}]\label{lemma:M->M'}
Every matrix in $\Fset^{n \times n}$ with rank at most $k$ and nonzero entries on the main diagonal has a principal sub-matrix that lies in $\calM_{n,k}^{(s')}$ for some $s'$.
\end{lemma}

Now, for integers $n,k,s'$, let $\calP_{n,k}^{(s')}$ be the collection that consists of all pairs $(M,R)$ such that, for some $n' \in [n]$, $M$ is an $n' \times n'$ matrix in $\calM_{n,k}^{(s')}$ and $R$ is an $n'$-subset of $[n]$.
Observe that Lemma~\ref{lemma:M->M'} implies that for every digraph $G$ on the vertex set $[n]$ with ${\minrank}_\Fset(G) \leq k$ there exist $s'$ and a pair $(M,R)$ in $\calP_{n,k}^{(s')}$ such that $M$ represents the induced subgraph $G[R]$ of $G$ on $R$, with respect to the natural order of the vertices in $R$ (from smallest to largest).

The following lemma provides an upper bound on the size of $\calP_{n,k}^{(s')}$.

\begin{lemma}\label{lemma:size_P}
For every integers $n,k,s'$, $|\calP_{n,k}^{(s')}| \leq (n \cdot |\Fset|)^{24s'k/n}$.
\end{lemma}

\begin{proof}
To bound the size of $\calP_{n,k}^{(s')}$, we consider for every $n' \in [n]$ and $k' \in [k]$ such that $\frac{k'}{n'} \leq \frac{k}{n}$ the pairs $(M,R)$ where $M$ is an $(n',k',s',\frac{2s'k'}{n'})$-matrix and $R$ is an $n'$-subset of $[n]$.
By Lemma~\ref{lemma:size_M} there are at most $(n' \cdot |\Fset|)^{12s'k'/n'}$ such matrices $M$, each of which occurs in $n \choose {n'}$ pairs of $\calP_{n,k}^{(s')}$. It follows that
\begin{eqnarray*}
|\calP_{n,k}^{(s')}| & \leq & \sum_{n',k'}{ {n \choose n'} \cdot (n' \cdot |\Fset|)^{12s'k'/n'}}
 \leq  n^2 \cdot \max_{n',k'} \big ( n^{n'} \cdot (n' \cdot |\Fset|)^{12s'k'/n'} \big ) \\
& \leq & \max_{n',k'} \big ( n^{3n'} \cdot (n' \cdot |\Fset|)^{12s'k'/n'} \big )
 \leq \max_{n',k'} \big ( (n \cdot |\Fset|)^{3n'+12s'k'/n'} \big )\\
& \leq & \max_{n',k'} \big ( (n \cdot |\Fset|)^{12s'k'/n' +12s'k'/n'} \big )
\leq  (n \cdot |\Fset|)^{24s'k/n},
\end{eqnarray*}
where in the fifth inequality we have used the relation $s' \geq \frac{n'^2}{4k'}$ from Lemma~\ref{lemma:sparsity_M}, and in the sixth we have used $\frac{k'}{n'} \leq \frac{k}{n}$.
\end{proof}

\subsection{Proof of Theorem~\ref{thm:IntroComp}}

We prove the following theorem and then derive Theorem~\ref{thm:IntroComp}.
Recall that for a graph $H$ with $h \geq 3$ vertices and $f \geq 3$ edges, we denote $\gamma(H) = \frac{h-2}{f-1}$.
We also let $\exp(x)$ stand for $e^x$.

\begin{theorem}\label{thm:Comp}
For every graph $H$ with at least $3$ edges there exists $c=c(H)>0$ such that for every integer $n$ and a finite field $\Fset$,
\[g(n,H,\Fset) \geq c \cdot \frac{n^{1-\gamma(H)}}{\log (n \cdot |\Fset|)} .\]
\end{theorem}

\begin{proof}
Fix a graph $H$ with $h \geq 3$ vertices and $f \geq 3$ edges and denote $\gamma = \gamma(H) = \frac{h-2}{f-1} > 0$.
The proof is via the probabilistic method. Let $\vec{G} \sim \vec{G}(n,p)$ be a random digraph on the vertex set $[n]$ where each directed edge is taken randomly and independently with probability $p$. Set $q=1-p$.
Let $G$ be the (undirected) graph on $[n]$ in which two distinct vertices $i,j$ are adjacent if both the directed edges $(i,j)$ and $(j,i)$ are included in $\vec{G}$. Notice that every two distinct vertices are adjacent in $G$ with probability $p^2$ independently of the adjacencies between other vertex pairs.

To prove the theorem, we will show that for a certain choice of $p$ the random graph $G$ satisfies with positive probability that its complement $\overline{G}$ is $H$-free and that ${\minrank}_{\Fset}(G) > k$, where
\begin{eqnarray}\label{eq:k}
k = c_1 \cdot \frac{n^{1-\gamma}}{\ln{(n \cdot |\Fset|)}}
\end{eqnarray}
for a constant $c_1>0$ that depends only on $H$.
To do so, we define two families of events as follows.

First, for every set $I \subseteq [n]$ of size $|I|=h$, let $A_I$ be the event that the induced subgraph of $\overline{G}$ on $I$ contains a copy of $H$. Observe that
\[\Prob{}{A_I} \leq h! \cdot (1-p^2)^f = h! \cdot (1-(1-q)^2)^f \leq h! \cdot (2q)^f.\]

Second, consider the collection $\calP = \cup_{s' \in [n^2]}{\calP_{n,k}^{(s')}}$ (see Section~\ref{sec:GRW}).
Recall that every element of $\calP$ is a pair $(M,R)$ such that, for some $n' \in [n]$, $M$ is an $n' \times n'$ matrix over $\Fset$ and $R$ is an $n'$-subset of $[n]$.
Denote $N_{s'} = |\calP_{n,k}^{(s')}|$.
By Lemma~\ref{lemma:size_P}, combined with~\eqref{eq:k}, we have
\begin{eqnarray}\label{eq:N_s'}
N_{s'} \leq (n \cdot |\Fset|)^{24s'k/n} = \exp(24c_1 \cdot s' \cdot n^{-\gamma}).
\end{eqnarray}
Let $\calS = \{s' \in [n^2] \mid N_{s'} \geq 1\}$.
By Lemma~\ref{lemma:sparsity_M}, for every $s' \in \calS$ and an $n' \times n'$ matrix of rank $k'$ in $\calM_{n,k}^{(s')}$ where $n' \in [n]$, $k' \in [k]$, and $\frac{k'}{n'} \leq \frac{k}{n}$, we have that
\begin{eqnarray}\label{eq:s'}
s' \geq \frac{n'}{4} \cdot \frac{n'}{k'} \geq \frac{n'}{4} \cdot \frac{n}{k} = n' \cdot \frac{n^{\gamma} \cdot \ln (n \cdot |\Fset|)}{4c_1}. \end{eqnarray}
Now, for every pair $(M,R) \in \calP$, let $B_{M,R}$ be the event that the matrix $M$ represents over $\Fset$ the induced subgraph $\vec{G}[R]$ of $\vec{G}$ on $R$ with respect to the natural order of the vertices in $R$.
For $M$ to represent $\vec{G}[R]$ we require that for every distinct $i, j$ such that $M_{i,j} \neq 0$, there is an edge in $\vec{G}$ from the $i$th to the $j$th vertex of $R$. Hence, for $M \in \Fset^{n' \times n'}$ of sparsity $s'$ and an $n'$-subset $R$ of $[n]$,
\[\Prob{}{B_{M,R}} = p^{s'-n'} \leq p^{s'/2} = (1-q)^{s'/2} \leq \exp(-qs'/2),\] where for the first inequality we have used the inequality $s' \geq 2n'$ which follows from~\eqref{eq:s'} for every sufficiently large $n$.

We claim that it suffices to prove that with positive probability none of the events $A_I$ and $B_{M,R}$ holds.
Indeed, this implies that there exists an $n$-vertex digraph $\vec{G}$ that does not satisfy any of these events.
Since the $A_I$'s are not satisfied it immediately follows that the complement $\overline{G}$ of the (undirected) graph $G$ associated with $\vec{G}$ is $H$-free.
We further claim that ${\minrank}_\Fset (G) > k$.
To see this, assume by contradiction that there exists a matrix $M \in \Fset^{n \times n}$ of rank at most $k$ that represents $G$, and thus, in particular, represents $\vec{G}$.
By Lemma~\ref{lemma:M->M'}, such an $M$ has a principal $n' \times n'$ sub-matrix $M' \in \calM_{n,k}^{(s')}$ for some $n'$ and $s'$.
Hence, for some $n'$-subset $R$ of $[n]$, the matrix $M'$ represents $\vec{G}[R]$ with respect to the natural order of the vertices in $R$, in contradiction to the fact that the event $B_{M',R}$ with $(M',R) \in \calP$ does not hold.

To prove that with positive probability none of the events $A_I$ and $B_{M,R}$ holds, we apply the Lov\'{a}sz Local Lemma (Lemma~\ref{lemma:lll}).
To this end, construct a (symmetric) dependency digraph $D=(V,E)$ whose vertices represent all the events $A_I$ and $B_{M,R}$, and whose edges are defined as follows.
\begin{itemize}
  \item An $A_I$-vertex and an $A_{I'}$-vertex are joined by edges (in both directions) if $|I \cap I'| \geq 2$. Notice that the events $A_I$ and $A_{I'}$ are independent when $|I \cap I'| < 2$.
  \item An $A_I$-vertex and a $B_{M,R}$-vertex are joined by edges if there are distinct $i,j \in I \cap R$ for which the entry of $M$ that corresponds to the edge $(i,j)$ is nonzero. Notice that the events $A_I$ and $B_{M,R}$ are independent when such $i$ and $j$ do not exist.
  \item Every two distinct $B_{M,R}$-vertices are joined by edges.
\end{itemize}
Clearly, each event is mutually independent of all other events besides those adjacent to it in $D$, and thus $D$ is a dependency digraph for our events.
Observe that every $A_I$-vertex is adjacent to at most ${h \choose 2} \cdot {n \choose {h-2}} \leq {h \choose 2} \cdot n^{h-2}$ $A_{I'}$-vertices.
Additionally, every $B_{M,R}$-vertex, where $M$ is an $n' \times n'$ matrix of sparsity $s'$, is adjacent to at most $(s'-n') \cdot {n \choose {h-2}} < s' \cdot n^{h-2}$ $A_{I}$-vertices.
Finally, every vertex of $D$ is adjacent to at most $N_{s'}$ $B_{M,R}$-vertices with $M \in \calM_{n,k}^{(s')}$ (that is, $s(M) = s'$).

To apply Lemma~\ref{lemma:lll} we assign a number in $[0,1)$ to each vertex of $D$.
Define
\[ q = c_2 \cdot n^{-\gamma},~~~~x = c_3 \cdot n^{-\gamma \cdot f},~~~~\mbox{and}~~~~x_{s'} = \exp(-c_4 \cdot s' \cdot n^{-\gamma})~~~~\mbox{for every $s' \in \calS$},\]
where $c_2,c_3,c_4>0$ are constants, depending only on $H$, to be determined.
We assign the number $x$ to every $A_I$-vertex, and the number $x_{s'}$ to every $B_{M,R}$-vertex with $s(M)=s'$.
We present now the conditions of Lemma~\ref{lemma:lll}.
For every $A_I$-vertex, recalling that $\Prob{}{A_I} \leq h! \cdot (2q)^f$, we require
\begin{eqnarray}\label{eq:lll_A}
h! \cdot (2q)^f \leq x \cdot (1-x)^{{h \choose 2} \cdot n^{h-2}} \cdot \prod_{s' \in \calS}{(1-x_{s'})^{N_{s'}}}.
\end{eqnarray}
Similarly, for every $B_{M,R}$-vertex with $s(M)=s'$, recalling that $\Prob{}{B_{M,R}} \leq \exp(-qs'/2)$, we require
\begin{eqnarray}\label{eq:lll_B}
\exp(-qs'/2) \leq x_{s'} \cdot (1-x)^{s' \cdot n^{h-2}} \cdot \prod_{s' \in \calS}{(1-x_{s'})^{N_{s'}}}.
\end{eqnarray}

To complete the proof, it suffices to show that the constants $c_1,c_2,c_3,c_4>0$ can be chosen in a way that satisfies the inequalities~\eqref{eq:lll_A} and~\eqref{eq:lll_B}. Consider the following three constraints:
\begin{enumerate}
  \item\label{itm:1} $c_2 > 2 \cdot (2c_3+c_4)$,
  \item\label{itm:2} $c_3 \geq h! \cdot (2c_2)^f \cdot \exp(3)$, and
  \item\label{itm:3} $c_4 \geq 32 \cdot c_1$.
\end{enumerate}
It is easy to see that it is possible to choose the constants under the above constraints. Indeed, by $f \geq 3$, for a sufficiently small choice of $c_2>0$ one can take $c_3$ with, say, an equality in Item~\ref{itm:2} so that some $c_4>0$ satisfies Item~\ref{itm:1}. Then, $c_1$ can be chosen as a positive constant satisfying Item~\ref{itm:3}.
We show now that such a choice satisfies~\eqref{eq:lll_A} and~\eqref{eq:lll_B} for every sufficiently large $n$. Note that we use below several times the inequality $1-\alpha \geq \exp(-2\alpha)$, which holds for any $\alpha \in [0,1/2]$.

First, use~\eqref{eq:N_s'} and the condition $c_4 \geq 32 \cdot c_1$ to obtain that
\[ \sum_{s' \in \calS}{x_{s'} \cdot N_{s'}} \leq \sum_{s' \in \calS}{\exp((24c_1-c_4) \cdot s' \cdot n^{-\gamma})} \leq \sum_{s' \in \calS}{\exp(-8c_1\cdot s' \cdot n^{-\gamma})} \leq
\sum_{s' \in \calS}{\exp(-2\ln n)} \leq 1, \]
where the third inequality follows by $s' \geq \frac{ n^{\gamma} \cdot \ln (n \cdot |\Fset|)}{4c_1}$ which we get from~\eqref{eq:s'}, and the fourth by $| \calS| \leq n^2$.
Considering the term $\prod_{s' \in \calS}{(1-x_{s'})^{N_{s'}}}$, which appears in both~\eqref{eq:lll_A} and~\eqref{eq:lll_B},
we derive that
\[ \prod_{s' \in \calS}{(1-x_{s'})^{N_{s'}}} \geq \prod_{s' \in \calS}{\exp(-2x_{s'} \cdot N_{s'})} = \exp \Big (-2 \cdot \sum_{s' \in \calS}{x_{s'} \cdot N_{s'}} \Big ) \geq \exp(-2).\]
For inequality~\eqref{eq:lll_A}, observe that
\begin{eqnarray*}
x \cdot (1-x)^{{h \choose 2} \cdot n^{h-2}} \cdot \prod_{s' \in \calS}{(1-x_{s'})^{N_{s'}}}
& \geq & x \cdot \exp \Big ( -2x \cdot {h \choose 2} \cdot n^{h-2} \Big ) \cdot \exp(-2) \\
& = & c_3 \cdot n^{-\gamma \cdot f} \cdot \exp \Big (-2c_3 \cdot n^{-\gamma \cdot f} \cdot {h \choose 2} \cdot n^{h-2} -2 \Big) \\
& \geq & h! \cdot (2c_2)^f \cdot n^{-\gamma \cdot f} \cdot \exp \Big(1-2c_3 \cdot {h \choose 2} \cdot n^{-\gamma} \Big) \\
& \geq & h! \cdot (2q)^f,
\end{eqnarray*}
where for the second inequality we use $c_3 \geq h! \cdot (2c_2)^f \cdot \exp(3)$ and $\gamma = \frac{h-2}{f-1}$,
and for the third we use the assumption that $n$ is sufficiently large.
For inequality~\eqref{eq:lll_B}, observe that
\begin{eqnarray*}
x_{s'} \cdot (1-x)^{s' \cdot n^{h-2}} \cdot \prod_{s' \in \calS}{(1-x_{s'})^{N_{s'}}}
& \geq & x_{s'} \cdot \exp (-2x \cdot s' \cdot n^{h-2} ) \cdot \exp (-2) \\
& = & \exp(-c_4 \cdot s' \cdot n^{-\gamma}) \cdot \exp (-2 c_3 \cdot n^{-\gamma \cdot f} \cdot s' \cdot n^{h-2} ) \cdot \exp (-2) \\
& = & \exp ( -(2c_3+c_4) \cdot s' \cdot n^{-\gamma} -2) \\
& \geq & \exp (-(c_2/2) \cdot s' \cdot n^{-\gamma}) \\
& = & \exp (-qs'/2),
\end{eqnarray*}
where for the second equality we again use the definition of $\gamma$, and for the second inequality we use the condition $c_2 > 2 \cdot (2c_3+c_4)$, the fact that $s' \cdot n^{-\gamma} = \omega(1)$ by~\eqref{eq:s'}, and the assumption that $n$ is sufficiently large. This completes the proof.
\end{proof}

We can derive now Theorem~\ref{thm:IntroComp}. Recall that $\gamma_0(H) = \min_{H'}{\gamma(H')}$, where the minimum is over all subgraphs $H'$ of $H$ with at least $3$ edges.

\begin{proof}[ of Theorem~\ref{thm:IntroComp}]
For a graph $H$ with $h \geq 3$ vertices and $f \geq 3$ edges, let $H'$ be a subgraph of $H$ with at least $3$ edges such that $\gamma_0(H) = \gamma(H')$.
By Theorem~\ref{thm:Comp} there exists $c>0$ such that
\[g(n,H',\Fset) \geq c \cdot \frac{n^{1-\gamma_0(H)}}{\log (n \cdot |\Fset|)}\]
for every integer $n$ and a finite field $\Fset$. Since every $H'$-free graph is also $H$-free, it follows that $g(n,H,\Fset) \geq g(n,H',\Fset)$ and we are done.
\end{proof}

\subsection{The Minrank of Graphs with Small Independence Number}

For an integer $t \geq 3$, $g(n,K_t,\Fset)$ is the maximum possible minrank over $\Fset$ of an $n$-vertex graph with independence number smaller than $t$. For this case we derive the following corollary.
\begin{corollary}\label{cor:K_t}
For every $t \geq 3$ there exists $c=c(t)>0$ such that for every integer $n$ and a finite field $\Fset$,
\[g(n,K_t,\Fset) \geq c \cdot \frac{n^{1-\frac{2}{t+1}}}{\log (n \cdot |\Fset|)} .\]
\end{corollary}

\begin{proof}
Apply Theorem~\ref{thm:IntroComp} to the graph $H = K_t$, and notice that $\gamma_0(K_t) = \gamma(K_t) = \frac{t-2}{{t \choose 2}-1} = \frac{2}{t+1}$.
\end{proof}

For $H=K_3$, we observe that our lower bound on $g(n,K_3,\Fset)$ is nearly tight.
\begin{proposition}\label{prop:K_3}
There exist constants $c_1,c_2>0$ such that for every integer $n$ and a finite field $\Fset$,
\[ c_1 \cdot \frac{\sqrt{n}}{\log (n \cdot |\Fset|)} \leq g(n,K_3,\Fset) \leq  c_2 \cdot \sqrt{\frac{n}{\log n}}.\]
\end{proposition}

\begin{proof}
For the lower bound apply Corollary~\ref{cor:K_t} with $t=3$.
To prove the upper bound we need a result of Ajtai et al.~\cite{AjtaiKS80} which says that every triangle-free $n$-vertex graph has an independent set of size $\Omega(\sqrt{n \cdot \log n})$. By repeatedly omitting such independent sets it follows that the chromatic number of such a graph is $O(\sqrt{n / \log n})$.
Now, let $G$ be an $n$-vertex graph whose complement $\overline{G}$ is triangle-free.
We get that ${\minrank}_{\Fset}(G) \leq \chi(\overline{G}) \leq O(\sqrt{n/\log n})$, as required.
\end{proof}

\section{Non-bipartite Graphs}\label{sec:non-bip}

In this section we show that for every non-bipartite graph $H$ there are $H$-free graphs with low minrank over $\R$, confirming Theorem~\ref{thm:IntroNonBi}.
We start with the case where $H$ is an odd cycle, and since every non-bipartite graph contains an odd cycle the general result follows easily.
The proof is by an explicit construction from the following family of graphs.

\begin{definition}\label{def:Kneser}
For integers $m \leq s \leq d$, the graph $\Kneser{d}{s}{m}$ is defined as follows: the vertices are all the $s$-subsets of $[d]$, and two distinct sets $A,B$ are adjacent if $|A \cap B| < m$.
\end{definition}

The minrank of such graphs over finite fields was recently studied in~\cite{Haviv18} using tools from~\cite{AlonBS91}.
The proof technique of~\cite{Haviv18} can be used for the real field as well, as shown below.

\begin{proposition}\label{prop:minrk_Kneser}
For every integers $m \leq s \leq d$,
\[{\minrank}_{\R}(\Kneser{d}{s}{m}) \leq \sum_{i=0}^{s-m}{d \choose i}.\]
\end{proposition}

\begin{proof}
Let $f: \{0,1\}^d \times \{0,1\}^d \rightarrow \R$ be the function defined by
\[ f(x,y) = \prod_{j=m}^{s-1}{ \Big ( \sum_{i=1}^{d}{x_i y_i -j}\Big )}\]
for every $x,y \in \{0,1\}^d$.
Expanding $f$ as a linear combination of monomials, the relation $z^2 = z$ for $z \in \{0,1\}$ implies that one can reduce to $1$ the exponent of each variable occuring in a monomial. It follows that $f$ can be represented as a multilinear polynomial in the $2d$ variables of $x$ and $y$. By combining terms involving the same monomial in the variables of $x$, one can write $f$ as
\[ f(x,y) = \sum_{i=1}^{R}{g_i(x) h_i(y)} \]
for an integer $R$ and functions $g_i, h_i : \{0,1\}^d \rightarrow \R$, $i \in [R]$, such that the $g_i$'s are distinct multilinear monomials of total degree at most $s-m$ in $d$ variables. It follows that $R \leq \sum_{i=0}^{s-m}{d \choose i}$.

Now, let $M_1$ and $M_2$ be the $2^d \times R$ matrices whose rows are indexed by $\{0,1\}^d$ and whose columns are indexed by $[R]$, defined by $(M_1)_{x,i} = g_i(x)$ and $(M_2)_{x,i} = h_i(x)$. Then, the matrix $M = M_1 \cdot M_2^T$ has rank at most $R$ and for every $x,y \in\{0,1\}^d$ it holds that $M_{x,y} = f(x,y)$.

Finally, let $V$ be the vertex set of $\Kneser{d}{s}{m}$, that is, the collection of all $s$-subsets of $[d]$, and identify every vertex $A \in V$ with an indicator vector $c_A \in \{0,1\}^d$ in the natural way. We claim that the matrix $M$ restricted to $V \times V$ represents the graph $\Kneser{d}{s}{m}$. Indeed, for every $A,B \in V$ we have
\[M_{c_A, c_B} = f(c_A,c_B) = \prod_{j=m}^{s-1}{ \Big ({|A \cap B| -j}\Big )}.\]
Hence, for every $A \in V$ we have $|A|=s$ and thus $M_{c_A,c_A} \neq 0$, whereas for every distinct non-adjacent $A,B \in V$ we have $m \leq |A \cap B|\leq s-1$ and thus $M_{c_A,c_B} = 0$. Since the restriction of $M$ to $V \times V$ has rank at most $R$ it follows that ${\minrank}_\R(\Kneser{d}{s}{m}) \leq R$, and we are done.
\end{proof}

We turn to identify graphs $\Kneser{d}{s}{m}$ with no short odd cycles.
For this purpose, take an even integer $d$, $s = \frac{d}{2}$, and $m = \eps \cdot d$ for a small constant $\eps>0$.
Every path in these graphs is a sequence of $\frac{d}{2}$-subsets of $[d]$ such that the intersection size of every two consecutive sets is small. This implies, for a sufficiently small $\eps$, that the sets in the even positions of the path are almost disjoint from the first set, whereas the sets in the odd positions of the path share with it many elements, hence such a graph contains no short odd cycle.
This is shown formally in the following lemma.

\begin{lemma}\label{lemma:cycle_K}
Let $\ell \geq 3$ be an odd integer.
For every even integer $d$ and an integer $m \leq \frac{d}{2\ell}$, the graph $\Kneser{d}{\frac{d}{2}}{m}$ contains no odd cycle of length at most $\ell$.
\end{lemma}

\begin{proof}
Fix an odd integer $\ell \geq 3$, an even integer $d$, and an integer $m \leq \frac{d}{2\ell}$.
We prove that for every odd integer $\ell'$, such that $3 \leq \ell' \leq \ell$, the graph $\Kneser{d}{\frac{d}{2}}{m}$ contains no cycle of length $\ell'$.
For such an $\ell'$, let $A_1,A_2,\ldots,A_{\ell'}$ be a sequence of $\ell'$ vertices in the graph, i.e., $\frac{d}{2}$-subsets of $[d]$. Assuming that for every $i \leq \ell'-1$ the vertices $A_i$ and $A_{i+1}$ are adjacent in the graph, that is, $|A_i \cap A_{i+1}| < m$, our goal is to show that $A_1$ and $A_{\ell'}$ are not.

To this end, we argue that for every $i$, such that $0 \leq i \leq \frac{\ell'-1}{2}$, we have
\begin{eqnarray}\label{eq:A_i}
|A_1 \cap A_{2i+1}| \geq \frac{d}{2}-2i \cdot m.
\end{eqnarray}
We prove this claim by induction on $i$. The case $i=0$ follows immediately from $|A_1|=\frac{d}{2}$.
Assume that~\eqref{eq:A_i} holds for $i-1$, that is, $|A_1 \cap A_{2i-1}| \geq \frac{d}{2}-(2i-2) \cdot m$.
Observe that this implies that
\begin{eqnarray*}
|A_1 \cap A_{2i}| &=& | A_1 \cap A_{2i} \cap A_{2i-1} | + | A_1 \cap A_{2i} \cap \overline{A_{2i-1}} | \\
& \leq & |A_{2i-1} \cap A_{2i}| + | A_1 \cap \overline{A_{2i-1}} | \\
& \leq & m + |A_1|- | A_1 \cap A_{2i-1} | \\
& \leq & m + \frac{d}{2} - \Big ( \frac{d}{2}-(2i-2) \cdot m \Big ) = (2i-1) \cdot m,
\end{eqnarray*}
where in the second inequality we have used $|A_{2i-1} \cap A_{2i}| < m$.
We proceed by proving~\eqref{eq:A_i} for $i$. Observe that
\begin{eqnarray*}
|A_1 \cap A_{2i+1}| &=& |A_{2i+1}| - | \overline{A_1} \cap A_{2i+1} | \\
&=& |A_{2i+1}| - | \overline{A_1} \cap A_{2i+1} \cap A_{2i} | - | \overline{A_1} \cap A_{2i+1} \cap \overline{A_{2i}} | \\
& \geq & \frac{d}{2} - m - | \overline{A_1} \cap \overline{A_{2i}} |,
\end{eqnarray*}
where we have used $|A_{2i} \cap A_{2i+1}| < m$.
Notice that
\[ |\overline{A_1} \cap \overline{A_{2i}}| = d - |A_1 \cup A_{2i}| = d-(|A_1|+|A_{2i}|-|A_1 \cap A_{2i}|) = |A_1 \cap A_{2i}|. \]
It follows that
\[|A_1 \cap A_{2i+1}| \geq \frac{d}{2}-m-|A_1 \cap A_{2i}| \geq \frac{d}{2}-m-(2i-1)\cdot m = \frac{d}{2}-2i\cdot m,\]
completing the proof of~\eqref{eq:A_i}.

Finally, applying~\eqref{eq:A_i} to $i = \frac{\ell'-1}{2}$, using the assumption $m \leq \frac{d}{2\ell}$, we get that
\[|A_1 \cap A_{\ell'}|  \geq \frac{d}{2}-(\ell'-1) \cdot m = \frac{d}{2}-\ell' \cdot m+m \geq  \frac{d}{2}-\ell \cdot m +m \geq m,\]
hence $A_1$ and $A_{\ell'}$ are not adjacent in the graph $\Kneser{d}{\frac{d}{2}}{m}$. It thus follows that the graph contains no cycle of length $\ell'$, as desired.
\end{proof}

Equipped with Proposition~\ref{prop:minrk_Kneser} and Lemma~\ref{lemma:cycle_K}, we obtain the following.

\begin{theorem}\label{thm:Cycles}
For every odd integer $\ell \geq 3$ there exists $\delta = \delta(\ell) >0$ such that for every sufficiently large integer $n$, there exists an $n$-vertex graph $G$ with no odd cycle of length at most $\ell$ such that
\[{\minrank}_{\R}(G) \leq n^{1-\delta}.\]
\end{theorem}

\begin{proof}
Fix an odd integer $\ell \geq 3$.
For an integer $d$ divisible by $2 \ell$, consider the graph $G = \Kneser{d}{\frac{d}{2}}{m}$ where $m = \frac{d}{2 \ell}$.
By Lemma~\ref{lemma:cycle_K}, $G$ contains no odd cycle of length at most $\ell$.
As for the minrank, Proposition~\ref{prop:minrk_Kneser} implies that
\[{\minrank}_{\R}(G) \leq \sum_{i=0}^{d/2-m}{d \choose i} \leq 2^{H(\frac{1}{2}-\frac{m}{d}) \cdot d} = 2^{H(\frac{1}{2}-\frac{1}{2\ell}) \cdot d},\]
where $H$ stands for the binary entropy function.
Since $G$ has $|V| = {d \choose {d/2}} = 2^{(1-o(1)) \cdot d}$ vertices, for any $\delta>0$ such that $H(\frac{1}{2}-\frac{1}{2\ell}) < 1-\delta$ we have ${\minrank}_{\R}(G) \leq |V|^{1-\delta}$ for every sufficiently large integer $d$.
The proof is completed by considering, for every sufficiently large integer $n$, some $n$-vertex subgraph of the graph defined above, where $d$ is the smallest integer divisible by $2\ell$ such that $n \leq {d \choose {d/2}}$.
\end{proof}

%\begin{remark}
%We note that known upper bounds on certain Ramsey numbers~\cite{ErdosFRS78} imply that for every $\ell \geq 3$, every $n$-vertex $C_\ell$-free graph has an independent set of size $\Omega( n^{1-1/k} )$ for $k = \lceil \frac{\ell}{2} \rceil$ (see~\cite{CaroLRZ00,Sudakov02} for slight improvements).
%By repeatedly omitting such independent sets it follows that the chromatic number of such a graph is $O(n^{1/k})$.
%This implies that every $n$-vertex graph $G$ whose complement is $C_\ell$-free satisfies ${\minrank}_{\Fset}(G) \leq \chi(\overline{G}) \leq O(n^{1/k})$, hence $g(n,C_\ell,\Fset) \leq O(n^{1/k})$.
%\end{remark}

Now, Theorem~\ref{thm:IntroNonBi} follows easily from Theorem~\ref{thm:Cycles}.

\begin{proof}[ of Theorem~\ref{thm:IntroNonBi}]
Let $H$ be a non-bipartite graph. Then, for some odd integer $\ell \geq 3$, the cycle $C_\ell$ is a subgraph of $H$.
By Theorem~\ref{thm:Cycles}, there exists $\delta > 0$ such that for every sufficiently large integer $n$, there exists an $n$-vertex $C_\ell$-free graph $G$ satisfying ${\minrank}_{\R}(G) \leq n^{1-\delta}$.
Since every $C_\ell$-free graph is also $H$-free, the result follows.
\end{proof}

\begin{remark}
As mentioned in the introduction, Theorem~\ref{thm:IntroNonBi} implies a lower bound on $g(n,H,\R)$ for every non-bipartite graph $H$ (see Corollary~\ref{cor:IntroNonBi}).
We note that upper bounds on certain Ramsey numbers can be used to derive upper bounds on $g(n,H,\Fset)$ for a general field $\Fset$.
For example, it was shown in~\cite{ErdosFRS78}
%We note that known upper bounds on certain Ramsey numbers~\cite{ErdosFRS78} imply 
that for every $\ell \geq 3$, every $n$-vertex $C_\ell$-free graph has an independent set of size $\Omega( n^{1-1/k} )$ for $k = \lceil \frac{\ell}{2} \rceil$ (see~\cite{CaroLRZ00,Sudakov02} for slight improvements).
By repeatedly omitting such independent sets it follows that the chromatic number of such a graph is $O(n^{1/k})$.
This implies that every $n$-vertex graph $G$ whose complement is $C_\ell$-free satisfies ${\minrank}_{\Fset}(G) \leq \chi(\overline{G}) \leq O(n^{1/k})$, hence $g(n,C_\ell,\Fset) \leq O(n^{1/k})$.
\end{remark}

\section*{Acknowledgements}
We are grateful to Alexander Golovnev and Pavel Pudl\'ak for useful discussions and to the anonymous referees for their valuable suggestions.

%\newpage
\bibliographystyle{abbrv}
\bibliography{minrk_free}

\begin{thebibliography}{10}

\bibitem{AjtaiKS80}
M.~Ajtai, J.~Koml{\'{o}}s, and E.~Szemer{\'{e}}di.
\newblock A note on {R}amsey numbers.
\newblock {\em J. Comb. Theory, Ser. {A}}, 29(3):354--360, 1980.

\bibitem{Alon94}
N.~Alon.
\newblock Explicit {R}amsey graphs and orthonormal labelings.
\newblock {\em Electr. J. Comb.}, 1(R12), 1994.

\bibitem{AlonPowers02}
N.~Alon.
\newblock Graph powers.
\newblock In B.~Bollob\'as, editor, {\em Contemporary Combinatorics}, Bolyai
  Society Mathematical Studies, pages 11--28. Springer, 2002.

\bibitem{AlonBS91}
N.~Alon, L.~Babai, and H.~Suzuki.
\newblock Multilinear polynomials and {F}rankl--{R}ay-{C}haudhuri--{W}ilson
  type intersection theorems.
\newblock {\em J. Comb. Theory, Ser. {A}}, 58(2):165--180, 1991.

\bibitem{AlonKS05}
N.~Alon, M.~Krivelevich, and B.~Sudakov.
\newblock Maxcut in ${H}$-free graphs.
\newblock {\em Combinatorics, Probability and Computing}, 14(5-6):629–--647,
  2005.

\bibitem{AlonS16}
N.~Alon and J.~H. Spencer.
\newblock {\em The Probabilistic Method}.
\newblock Wiley Publishing, 4th edition, 2016.

\bibitem{BBJK06}
Z.~Bar-Yossef, Y.~Birk, T.~S. Jayram, and T.~Kol.
\newblock Index coding with side information.
\newblock In {\em FOCS}, pages 197--206, 2006.

\bibitem{BlasiakKL13}
A.~Blasiak, R.~Kleinberg, and E.~Lubetzky.
\newblock Broadcasting with side information: Bounding and approximating the
  broadcast rate.
\newblock {\em {IEEE} Trans. Information Theory}, 59(9):5811--5823, 2013.

\bibitem{CaroLRZ00}
Y.~Caro, Y.~Li, C.~C. Rousseau, and Y.~Zhang.
\newblock Asymptotic bounds for some bipartite graph: complete graph {R}amsey
  numbers.
\newblock {\em Discrete Mathematics}, 220(1-3):51--56, 2000.

\bibitem{ChlamtacH14}
E.~Chlamt\'a\v{c} and I.~Haviv.
\newblock Linear index coding via semidefinite programming.
\newblock {\em Combinatorics, Probability {\&} Computing}, 23(2):223--247,
  2014.
\newblock Preliminary version in SODA'12.

\bibitem{CodenottiPR00}
B.~Codenotti, P.~Pudl{\'{a}}k, and G.~Resta.
\newblock Some structural properties of low-rank matrices related to
  computational complexity.
\newblock {\em Theor. Comput. Sci.}, 235(1):89--107, 2000.
\newblock Preliminary version in ECCC'97.

\bibitem{ErdosFRS78}
P.~Erd{\"{o}}s, R.~J. Faudree, C.~C. Rousseau, and R.~H. Schelp.
\newblock On cycle-complete graph {R}amsey numbers.
\newblock {\em J. Graph Theory}, 2(1):53--64, 1978.

\bibitem{LLL75}
P.~Erd{\"{o}}s and L.~Lov{\'{a}}sz.
\newblock Problems and results on 3-chromatic hypergraphs and some related
  questions.
\newblock In A.~Hajnal, R.~Rado, and V.~T. S\'os, editors, {\em Infinite and
  Finite Sets}, pages 609–--627. North-Holland, Amsterdam, 1975.

\bibitem{ErdosMT71}
P.~Erd{\"{o}}s, R.~J. McEliece, and H.~Taylor.
\newblock Ramsey bounds for graph products.
\newblock {\em Pacific J. Math.}, 37(1):45--46, 1971.

\bibitem{Golovnev0W17}
A.~Golovnev, O.~Regev, and O.~Weinstein.
\newblock The minrank of random graphs.
\newblock In {\em Randomization and Approximation Techniques in Computer
  Science ({RANDOM})}, pages 46:1--46:13, 2017.

\bibitem{Haemers79}
W.~Haemers.
\newblock On some problems of {L}ov\'asz concerning the {S}hannon capacity of a
  graph.
\newblock {\em IEEE Trans. Inform. Theory}, 25(2):231--232, 1979.

\bibitem{Haemers81}
W.~Haemers.
\newblock An upper bound for the {S}hannon capacity of a graph.
\newblock In {\em Algebraic methods in graph theory, {V}ol. {I}, {II}
  ({S}zeged, 1978)}, volume~25 of {\em Colloq. Math. Soc. J\'anos Bolyai},
  pages 267--272. North-Holland, Amsterdam, 1981.

\bibitem{Haviv18}
I.~Haviv.
\newblock On minrank and the {L}ov\'asz theta function.
\newblock In {\em Approximation Algorithms for Combinatorial Optimization
  Problems ({APPROX})}, 2018.
\newblock To appear.

\bibitem{HavivL13}
I.~Haviv and M.~Langberg.
\newblock {$H$}-wise independence.
\newblock In {\em Innovations in Theoretical Computer Science (ITCS'13)}, pages
  541--552, 2013.

\bibitem{KasKon81}
B.~S. Kashin and S.~V. Konyagin.
\newblock Systems of vectors in {H}ilbert space.
\newblock In {\em Number theory, mathematical analysis, and their
  applications}, volume 157 of {\em Trudy Mat. Inst. Steklov.}, pages 64--67.
  1981.

\bibitem{Kon81}
S.~V. Konyagin.
\newblock Systems of vectors in {E}uclidean space and an extremal problem for
  polynomials.
\newblock {\em Mat. Zametki}, 29(1):63--74, 1981.

\bibitem{Lovasz79}
L.~Lov{\'a}sz.
\newblock On the {S}hannon capacity of a graph.
\newblock {\em IEEE Trans. Inform. Theory}, 25(1):1--7, 1979.

\bibitem{Peeters96}
R.~Peeters.
\newblock Orthogonal representations over finite fields and the chromatic
  number of graphs.
\newblock {\em Combinatorica}, 16(3):417--431, 1996.

\bibitem{Pudlak02}
P.~Pudl{\'{a}}k.
\newblock Cycles of nonzero elements in low rank matrices.
\newblock {\em Combinatorica}, 22(2):321--334, 2002.

\bibitem{PudlakRS97}
P.~Pudl{\'{a}}k, V.~R{\"{o}}dl, and J.~Sgall.
\newblock Boolean circuits, tensor ranks, and communication complexity.
\newblock {\em {SIAM} J. Comput.}, 26(3):605--633, 1997.

\bibitem{Riis07}
S.~Riis.
\newblock Information flows, graphs and their guessing numbers.
\newblock {\em Electr. J. Comb.}, 14(1), 2007.

\bibitem{Rosenfeld91}
M.~Rosenfeld.
\newblock Almost orthogonal lines in ${E}^d$.
\newblock {\em DIMACS Series in Discrete Math.}, 4:489--492, 1991.

\bibitem{Shannon56}
C.~E. Shannon.
\newblock The zero error capacity of a noisy channel.
\newblock {\em Institute of Radio Engineers, Transactions on Information
  Theory}, IT-2:8--19, 1956.

\bibitem{Spencer77}
J.~Spencer.
\newblock Asymptotic lower bounds for {R}amsey functions.
\newblock {\em Discrete Mathematics}, 20:69--76, 1977.

\bibitem{Sudakov02}
B.~Sudakov.
\newblock A note on odd cycle-complete graph {R}amsey numbers.
\newblock {\em Electr. J. Comb.}, 9(1), 2002.

\bibitem{Valiant77}
L.~G. Valiant.
\newblock Graph-theoretic arguments in low-level complexity.
\newblock In {\em Mathematical Foundations of Computer Science (MFCS), 6th
  Symposium}, pages 162--176, 1977.

\bibitem{Valiant92}
L.~G. Valiant.
\newblock Why is {B}oolean complexity theory difficult?
\newblock In {\em Poceedings of the London Mathematical Society symposium on
  Boolean function complexity}, volume 169, pages 84--94, 1992.

\bibitem{XiaodongZER04}
X.~Xu, X.~Zheng, G.~Exoo, and S.~P. Radziszowski.
\newblock Constructive lower bounds on classical multicolor {R}amsey numbers.
\newblock {\em Electr. J. Comb.}, 11(1), 2004.

\end{thebibliography}

\end{document}